\documentclass[10pt]{article}
\usepackage[T1]{fontenc}
\usepackage{authblk}
\usepackage{amssymb}
\usepackage{amsmath}
\usepackage{babel}
\usepackage{subfig}
\usepackage{comment}
\usepackage[normalem]{ulem}
\usepackage[mathcal]{eucal}
\usepackage[]{graphicx}
\usepackage[round]{natbib}
\usepackage{tikz}
\usepackage{bm}
\newcommand{\bi}{\begin {itemize}}
\newcommand{\ei}{\end{itemize}}

\usepackage{amsthm}
\newtheorem{thm}{Theorem}
\newtheorem*{thm*}{Theorem}
\newtheorem{prop}{Proposition}
\newtheorem*{prop*}{Proposition}
\newtheorem{lem}{Lemma}
\newtheorem*{lem*}{Lemma}

\newtheorem{cor}{Corollary}
\newtheorem{ass}{Assumption}

\newcommand{\obsFlip}{F}
\newcommand{\flipSpace}{\mathbb{F}}
\newcommand{\II}{\mathbf{I}}

\newcommand{\flip }{\mathcal{F}}
\newcommand{\ones}{\mathbf{1}}

\newcommand{\norm}{||}

\newcommand*{\riccardo}[1]{\textcolor{black}{#1}}

\newcommand{\asV}{\tilde{V}}
\newcommand{\asW}{\tilde{W}}

\newcommand{\asA}{\tilde{a}}

\newcommand{\asphi}{\tilde{\phi}}

\newcommand{\asnu}{\tilde{\nu}}

\newcommand{\asH}{\tilde{H}}
\newcommand{\ash}{\tilde{h}}

\newcommand{\Var}{\mbox{var}}
\newcommand{\E}{E}

\newcommand{\estphi}{\hat{\phi}}
\newcommand{\estmu}{\hat{\mu}}

\usepackage{xcolor}

\title{Inference in generalized linear models with robustness to misspecified variances}
\author[1]{Riccardo De Santis\footnote{To contact:  \textit{riccardo.desantis2@unisi.it}}}
\author[2]{Jelle J. Goeman}
\author[3]{Jesse Hemerik}
\author[4]{Samuel Davenport}
\author[1]{Livio Finos}

\affil[1]{University of Padova, Italy} 
\affil[2] {Leiden University Medical Center, The Netherlands} 
\affil[3]{Erasmus University Rotterdam, The Netherlands}
\affil[4]{University of California San Diego, United States} 
\date{}

\begin{document}

\maketitle

\begin{abstract}
Generalized linear models usually assume a common dispersion parameter, an assumption that is seldom true in practice. Consequently, standard parametric methods may suffer appreciable loss of type I error control. As an alternative, we present a semi-parametric group-invariance test based on sign flipping of score contributions. Our method requires only the correct specification of the mean model, but is robust against any misspecification of the variance. We present tests for single as well as multiple regression coefficients. The test is asymptotically valid but shows excellent performance in small samples. We illustrate the method using RNA sequencing count data, for which it is difficult to model the overdispersion correctly.
The method is available in the {\tt R} library {\tt flipscores}. 
\end{abstract}

\section{Introduction}\label{section:intro}
When testing for equality of means between two groups, statisticians often worry that there may be a difference in variance between the two groups. The task of properly taking this variance into account, known as the Behrens-Fisher problem \citep{Fisher:1935, Fisher:1941}, has generated a huge statistical literature \citep{B-Freview:1998, B-Freview:2008}. \riccardo{When testing a coefficient in a regression model, the situation becomes more complicated;}
in the linear regression model (LM), which generalizes the two-group model, the assumption that all error terms have equal variance is less often questioned. While goodness-of-fit tests and other diagnostics for heteroscedasticity exist \citep{heterosk6, heterosk4, heterosk1, heterosk5, heterosk2, heterosk3}, there is no simple and general way to follow up on obvious lack of fit \citep{critics_test_het}. Still, the problems arising from misspecified variances in regression can be as severe in regression as they are for the two-group model, especially when the variance depends on the covariates in the model in an unknown way: standard errors are too large or too small, and statistical tests \riccardo{for the coefficients} can become severely conservative or anti-conservative.

The situation is worse if we broaden the perspective to generalized linear models (GLMs). Overdispersed GLMs, that allow for additional variation between subjects, have become important in many fields; e.g. negative binomial or quasi-Poisson models in RNA sequencing \citep{RnaSeq}.
GLMs and overdispersed GLMs generally use a single dispersion parameter for all subjects. This assumption is similar to the common variance assumption in the linear model, and  violation of the assumption can lead to loss of type I error control in the same way.
For large data sets this problem can be addressed by estimating the variance of the test statistic robustly using the sandwich estimator \citep{sandwich2,sandwich3} or by direct modelling using Generalized Additive Models for Location Scale and Shape \citep{Gamlss}. However, \riccardo{tests based on }these methods can have poor error control for small sample sizes  \citep{SandwichCritic1,SandwichCritic2, SandwichCritic3, SandwichCritic4}.


In linear models the wild bootstrap \citep{wildboot} addresses the problem of unknown heteroscedasticity in a different way by randomly sign-flipping transformed residuals. This method bypasses the need to model the variance by comparing each subject's contribution to the test statistic with its own sign-flipped copy only. The method of \cite{Hemerik.etal:2020} extends this method to generalized linear models and to situations with nuisance parameters by sign-flipping the contributions to the score statistic rather than the residuals. However, their contribution has two main drawbacks. First, the proven robustness of the methods to variance misspecification is rather limited, relating only to misspecification by a constant. Second, although it has better control of type I error than Sandwich-based methods, it still struggles to control type I error properly for small sample sizes. 


In this paper, we present an improvement of the method of \cite{Hemerik.etal:2020}, addressing both drawbacks. In the first place we develop a standardization of the sign-flip-based test statistics that boosts the convergence to the asymptotic distribution and greatly improves the control of the type one error. Our proposed standardization is innovative because the standardization factor is conditional on the random sign-flips and does not require model refitting as is common for resampling based standardized test statistics \citep[e.g., methods reviewed in][]{Winkler.et.al.2014, permuco}. Second, we prove robustness to any misspecification of the variance, greatly improving the scope of the method.

The structure of this paper is as follows. In Section \ref{section:main_results} we summarize the main results of the paper, placing them in the context of the literature. In Section \ref{section:glm} we introduce the modeling context and assumptions that we will consider. After recalling the existing score flipping test in Section \ref{section:effective},  in Section \ref{sect:standardized} we propose the novel test. Section \ref{sec:robust} contains the most important results of this paper, since here we prove the robustness properties of our tests. An extension to multivariate tests is in Section \ref{sect:multivariate}. Section \ref{sect:simulation} contains the simulation study and Section \ref{section:realdata} a data analysis, which illustrates the relevance of our approach to the analysis of RNA-Seq data.

\section{Contributions of the paper}\label{section:main_results}
We start by summarizing the two main novel contributions of this paper, placing them in the context of the literature. 

\subsection{Robustness to variance misspecification in GLMs}

While many methods have been proposed to remedy the problem variance misspecification, they all come at a price.

A general and popular approach that is robust to variance misspecification is the sandwich estimator \citep{cox1961tests, cox1962further,sandwich1,sandwich2,sandwich3}. However, this method is well-known to have poor type I error control for small sample sizes \citep{SandwichCritic1, SandwichCritic4, SandwichCritic3, SandwichCritic2}. Alternatively, more flexible models can be used to tackle the problem of the strict assumptions made by the GLMs. Generalized additive models for location, scale and shape \citep[GAMLSS;][]{Gamlss} permit to jointly model the mean and the variance as a function of covariates. However, these models tend to be overparametrized, which can also result in poor small sample performance, as we will see in the simulations of Section \ref{sect:simulation}. Furthermore, constructing and evaluating a variance model \citep[GAMLSS;][]{Gamlss} is more complex in GLMs than in the linear model, since the variance of the residuals differs between observations even in a correctly specified model.

A completely different perspective is taken by HulC (Hull based Confidence) \citep{kuchibhotla2023hulc}, a recent robust semiparametric method based on sample splitting, whose key assumption is only the asymptotic median unbiasedness of the estimator. This method gives confidence intervals with theoretical guarantees. However, as a sample splitting method, it requires a minimal sample size, and comes with a relevant loss of power. We will consider this method in the simulations of section \ref{sect:simulation}.

The Wild Bootstrap is highly robust to variance mispecification since it uses within-observation sign flips of the residuals, refitting to the resulting pseudo-observations. 
However, this method does not extend to the GLM context since the new pseudo-observations do not generally respect the original GLM model.

In this paper we adopt the related perspective of the sign-flipping test of \citet{Hemerik.etal:2020}.
In this approach, rather than the residuals, the score contributions are sign-flipped. For this approach, \citet{Hemerik.etal:2020} proved a very limited robustness to variance misspecification; only for variances that are misspecified up to a multiplicative constant. We prove that their test remains asymptotically valid under any unknown misspecification of the variance. 

\subsection{Standardization of the test statistics: marginal second-moment exactness}\label{Subsec:sec_mom_null_inv}

The score-flipping test \citep{Hemerik.etal:2020} that we use as our starting point, has a tendency to anticonservativeness for small sample sizes. We prove that this anticonservative tendency is always present when that test is applied to LMs with correctly specified variances, and observe it---in simulations---in general. To address this problem we develop a novel standardization approach to the test statistics which boosts its convergence to the nominal level. 
Uniquely, the proposed standardization is conditional on the random sign flips. This conditional standardization ensures marginal second moment exactness, a property not shared by other resampling-based methods \citep{Winkler.et.al.2014, permuco}.

To understand the philosophy of the new standardization approach, it is helpful to revisit the way permutation tests achieve exact control of the type I error. Let $H_0$ be the null hypothesis that we want to test. By exactness of a test we mean that if $H_0$ holds, the probability of rejecting $H_0$ is exactly $\alpha$. To get such exact control, one needs to ensure that, if the null hypothesis is true, the joint distribution of the test statistics $T(g Y),\ g \in G$, is invariant under all transformations in a group $G$ of the data $Y$.
That is \citep[Definition 1]{Hemerik.goeman:2018}, under the null hypothesis,
\begin{eqnarray} \label{eq:exact}
\big(T (g' Y) \big)_{g' \in G} \overset{d}{=}  \big(T (g' g Y) \big)_{g' \in G}
\end{eqnarray}
for all $g \in G$. A sufficient condition for this is that
$Y \overset{d}{=}  g Y$ for all $g \in G$. Once \eqref{eq:exact} is satisfied, tests can be based on randomly selected transformations $G'$ from $G$.

While (\ref{eq:exact}) is easy to achieve in simple experimental designs, it is difficult or impossible to achieve in more complex settings such as GLMs. Even within the LM framework, most known methods achieve only asymptotic exactness; that is, the distributions of
$\big(T (g' Y) \big)_{g' \in G'}$ and $\big(T (g' g Y) \big)_{g' \in G'}$
are not identical for all $g \in G$, but merely converge to the same (multivariate normal) distribution. 
As a result, type I error control is only asymptotic.

When exact equality in (\ref{eq:exact}) is impossible, a fallback option is to ensure marginal second-moment null-invariance \citep{Commenges2003}, i.e., to let the first and marginal second moments of the distributions in (\ref{eq:exact}) to be equal in finite samples. This way, asymptotic arguments are only needed for mixed and higher order moments, which tend to converge faster.
We achieve marginal second-moment null-invariance by standardizing each $T(g Y)$ by its standard deviation, which depends on $g$. To the best of our knowledge, we are the first to propose such a per-transformation standardization.  


\section{Set-up and assumptions}\label{section:glm}
Assume that we observe $n$ independent observations $y=\{y_1,\dots,y_n\}^T$ from the exponential dispersion family, i.e., 
a density of the form \citep{Agresti:2015}
\begin{equation} \label{eq:model}
    f(y_i; \theta_i, \phi_i)=
    \exp\left\{\frac{y_i\theta_i-b(\theta_i)}{a(\phi_i)}+c(y_i,\phi_i)\right\},
\end{equation}
where $\theta_i$ and $\phi_i$ are respectively the canonical and the dispersion parameter. \riccardo{Note that we do not make notational distinction between the random variable and the single realization, but the context will ensure no ambiguity.}
We will assume throughout the paper that the model chosen fulfills the usual regularity conditions \citep[Chapter 3]{Azzalini1996}.
We derive mean and variance of the observed outcome, respectively, as
\begin{equation*}
\mu_i=\E(y_i)=b'(\theta_i); \qquad
\Var(y_i)=b''(\theta_i)a(\phi_i),
\end{equation*}
where primes denote derivatives. Without loss of generality, let $a(\phi_i)=\phi_i$. We assume that the mean of $y_i$ depends on observed covariates $(x_i, z_i)$ through the following relation, written in vector-matrix form as
\begin{equation*}
    g(\mu)=\eta=X\beta+Z\gamma
\end{equation*}
where $\mu=(\mu_1,\dots,\mu_n)^T$ denotes the mean vector, $g(\cdot)$ is the link function, taken to operate elementwise on a vector, $(X,Z)$ is the full rank design matrix with $\mathrm{dim}(X)=n \times1$, $\mathrm{dim}(Z)=n \times q$, where $q$ does not depend on $n$, and $( \beta, \gamma)$ are unknown parameters. We consider $\beta$ the parameter of interest, and $\gamma$ as a nuisance parameter. \riccardo{We will treat the covariates as fixed through the conditioning implicit in the regression model.} 

Also $\phi_1, \ldots, \phi_n$, are all nuisance parameters; we leave the structure of the dispersion parameters completely unspecified. Dispersion parameters can have very different values due to a variety of causes, such as the omission of some covariates \citep{Agresti:2015}. These $n$ separate dispersion parameters cannot be consistently estimated unless they are known to satisfy some restrictions \citep{NeymanScott:1948}. 
Thus, in practice, the analyst will use some, possibly misspecified, putative model, which imposes such restrictions. This leads to computable estimates ${\estphi}_1,\dots,{\estphi}_n$, which will generally be inconsistent, since we cannot assume the putative model to be correct.

As an example of the set-up, consider the situation that the data are generated according to a negative binomial model with a log link function and a different and unknown dispersion parameter for every observation. The analyst could estimate these dispersion parameters under the additional constraint that $\phi_i = \phi$ for all $i$, or even that $\phi_i = 1$ for all $i$. Alternatively, the dispersion parameters could be modeled as a function of the covariates \citep[Chapter 10]{McCullagh1989}. Such strategies will lead to inconsistent estimates of the dispersion parameters unless the true model happens to fulfil the chosen constraints \citep{sandwich3}.

About the model and the estimation strategy we make the following assumptions.
\begin{ass}\label{ass:link}
    Let \eqref{eq:model} be the true model which generates the data. We assume that the link function $g(\cdot)$ is correctly specified.
\end{ass}
This assumption is crucial since it permits consistent estimation of the regression parameters $\beta, \gamma$, though not of $\phi_1,\ldots, \phi_n$. Indeed, the incorrect estimation of the dispersion parameters $\phi_i$ does not affect the consistency of the estimates of the regression coefficients, as long as the mean (and hence the link function) are well-specified \citep{Agresti:2015}. The corresponding standard errors, however, are no longer reliable.

Now, let $\ell(\beta, \gamma)$ denote the log-likelihood function, 
from which we can derive the score vector with elements
\begin{equation*}
    \frac{\partial \ell(\beta,\gamma)}{\partial \beta}=s_\beta=X^T D V^{-1}(y-\mu); \qquad  \frac{\partial \ell(\beta,\gamma)}{\partial \gamma}=s_\gamma=Z^T D V^{-1}(y-\mu),
\end{equation*}
where, in a compact matrix form, we have $$D=\mathrm{diag}\left\{\frac{\partial \mu_i}{\partial \eta_i}\right\}; \qquad V=\mathrm{diag}\{\mathrm{\Var}(y_i)\}.$$
Taking the derivatives of the score vector we obtain the Fisher information matrix
\begin{equation*}
    \mathcal{I}=
    \begin{pmatrix}
    \mathcal{I}_{\beta,\beta} & \mathcal{I}_{\beta,\gamma} \\
    \mathcal{I}_{\gamma,\beta} &
    \mathcal{I}_{\gamma,\gamma}
    \end{pmatrix}
    =
    \begin{pmatrix}
    X^T W X & 
    X^T W Z \\
    Z^T W X &
    Z^T W Z
    \end{pmatrix}
\end{equation*}
where $W=D V^{-1} D$. Denote by $d_i$, $v_i$, and $w_i$ the $i$-th diagonal elements of $D$, $V$ and $W$, respectively.
Note that all the matrices defined here and in the rest of the paper depend on $n$. We will suppress this dependence in the notation. The reader may assume that all quantities depend on $n$, except when explicitly stated otherwise.

Secondly, we assume that the estimates of $\phi_1, \phi_2, \ldots$, though not consistent, will converge, as stated more formally below. 

\begin{ass}\label{ass:dispersion}
For $i=1,2, \ldots$, we have $\hat\phi_i - \tilde\phi_i \to 0$ in probability as $n \to\infty$, where $K_1 < \tilde\phi_i \leq K_2$, and $K_1, K_2$ are strictly positive constants not depending on $n$. 
\end{ass}
It is known \citep{sandwich2} that the use of maximum likelihood estimation in a misspecified model leads under minimal conditions to a well-defined limit of the estimator. In this sense, it is possible to define a ``limit'' density, even in a misspecified model, which is intended to be the closest to the true density which generates the data in terms of Kullback-Leibler distance \citep{sandwich3}. 
This implies that Assumption 1 generally holds in situations where we estimate $\phi_1, \ldots, \phi_n$ using a restricted model.

We place a tilde symbol on all the quantities obtained when plugging the limits $\asphi_1,\ldots,\asphi_n$ into the model. In particular, let $\asV, \asW$ be the variance and weight matrices obtained by fixing $\phi_i = \asphi_i$. 

We further assume the following assumption related to the Fisher information given by the quantities $\asphi_1,\dots,\asphi_n$. This is a standard condition in regular models \citep{Vaart1998}.
\begin{ass}\label{ass:fisherinfo}
Let $\tilde{\mathcal{I}}$ be the Fisher information matrix for $W=\asW$. The $\lim_{n \xrightarrow{} \infty} n^{-1} \tilde{\mathcal{I}}_{\beta,\beta}$ converges to some positive constant.
\end{ass}

Further, the following mild condition is needed to apply the multivariate central limit theorem \citep[Chapter 5]{Billingsley1986} within the framework of \cite{Hemerik.etal:2020}, that we will apply. It is needed to avoid pathological cases, such as vanishing or dominating observations.
\begin{ass}  \label{ass:lindeberg}
Let $$\asnu_{i,\beta}=\frac{(y_i-\estmu_{i})x_id_{i}}{\tilde v_i}; \qquad \asnu_{i,\gamma}=\frac{(y_i-\estmu_{i})z_id_{i}}{\tilde v_i};$$
define the element-wise score contribution and
\begin{equation*}
    \asnu^*_{i,\beta}=\asnu_{i,\beta}-\tilde{\mathcal{I}}_{\beta,\gamma}\tilde{\mathcal{I}}_{\gamma,\gamma}^{-1}\asnu_{i,\gamma}.
\end{equation*}
We require, for all $\epsilon >0$,
\begin{equation*}
    \lim_{n \xrightarrow{} \infty} \frac{1}{n} \sum_{i=1}^n E[(\tilde{\nu}_{i,\beta}^{*})^2 \mathbf{1}_{\{|\tilde{\nu}_{i,\beta}^*|/\sqrt{n}>\epsilon\}}] \xrightarrow{} 0
\end{equation*}
where $\mathbf{1}_{\{\dots\}}$ is the indicator function, and $n^{-1}\sum_{i=1}^n \Var(\tilde{\nu}_{i,\beta}^{*})$ to converge to some positive constant.
\end{ass}

When the model is correctly specified the last condition implies Assumption \ref{ass:fisherinfo}. In case of misspecified models we do not have this relation, since the so-called ``information identity'' does not hold \citep{Azzalini1996}.

Throughout the paper, we will make a distinction between two situations: the situation that the model is correctly specified, i.e.\ $\asW = W$, and the putative model is true; or the general case that $\asW \neq W$.

\section{Sign-flipping effective score test}\label{section:effective}

In the model described in the previous section we are interested in testing the following null hypothesis
\begin{eqnarray} \label{H_0_true}
H_0:\ \beta= \beta_0\ |\ (\gamma, \phi_1, \dots, \phi_n) \in \Gamma \times \Phi \times \dots \times \Phi
\end{eqnarray}
against a one or two-sided alternative, where $\Gamma \subseteq \mathbb{R}^q$ and $\Phi \subseteq (0,\infty)$.
The formulation of the null hypothesis above makes explicit that only the target parameter $\beta$ is fixed under $H_0$, but the nuisance parameters are unconstrained. 

Since $\phi_1, \dots, \phi_n$, as remarked, are difficult to estimate, we require a test that is robust to misspecification of these parameters. \cite{Hemerik.etal:2020} proposed a general semi-parametric test that has robustness to misspecification of the variance by a constant, and presented promising simulation results suggesting more general robustness properties. We will start from this test.

\cite{Hemerik.etal:2020} used the effective score for $\beta$ as the test statistic. \riccardo{A general discussion on the definition of effective scores and its properties can be found in \citet[p.847]{Hemerik.etal:2020}, \citet[p.344]{marohn2002comment}, \citet[p.453]{murphy2000profile}, \citet[p.437]{begun1983information} or  \citet[p.146]{pace1997principles}.} In the model (\ref{eq:model}) the effective score is
\[
S = s_\beta - \mathcal{I}_{\beta, \gamma} \mathcal{I}_{\gamma, \gamma}^{-1} s_\gamma.
\] 
In the context of generalized linear models, the statistic can be written as
\begin{equation*}
    S=n^{-1/2}X^TW^{1/2}(I-H) V^{-1/2}(y-{\estmu})
\end{equation*}
where
\begin{equation}\label{eq:hproj}
    H=W^{1/2}Z(Z^TW Z)^{-1}Z^TW^{1/2}
\end{equation}
is the hat matrix, and $\estmu$ is the vector of fitted values of the model under the null hypothesis. Note that, since $S$ is an inner product of the two $n$-vectors $n^{-1/2}V^{-1/2}(I-H) W^{1/2}X$ and $y-{\estmu}$, it can be written as a sum of $n$ terms, which we call the effective score contributions.

To calculate the critical value, \cite{Hemerik.etal:2020} proposed to sign-flip the effective score contributions, randomly multiplying each score contribution by $-1$ or $1$.  In matrix notation these sign flips can be represented by a random diagonal matrix $\flip$ of dimension $n$, whose non-zero elements are independent random variables that take values $-1$ and 1 with equal probability.
Consequently, the effective sign-flip score statistic for a given flip matrix $\flip=\obsFlip$ is defined as
\begin{equation}\label{eq:eff}
    S(\obsFlip)=n^{-1/2}X^TW^{1/2}(I-H)V^{-1/2}\obsFlip (y-{\estmu}). 
\end{equation}
Note that for $\flip = \II$ (the identity matrix) we recover the observed effective score.

An asymptotic $\alpha$-level test is then derived as follows. First, \cite{Hemerik.etal:2020} prove (asymptotic) invariance of the first two moments of the test statistic under the action of $\flip$, i.e.\ that $E\{S(\flip)\} - E\{S(\II)\} = 0$ and $\Var\{S(\flip)\} - \Var\{S(\II)\} \to 0$ as $n \to\infty$. Next, they apply the Lindberg-Feller multivariate central limit theorem to show that, for independently drawn flip matrices $F_2, \ldots, F_g$, where $g$ does not depend on $n$, the vector $S(\II), S(F_2), \ldots, S(F_g)$ converges to a vector of independent and identically distributed random variables. Lemma 1 of \cite{Hemerik.etal:2020} then allows to obtain an asymptotic $\alpha$-level test for the null hypothesis \eqref{H_0_true} against a one or two-sided alternative. Abbreviate $S_i = S(F_i)$ and denote the corresponding sorted values as $S_{(1)}\le\dots\le S_{(g)}$. Without loss of generality, consider testing \eqref{H_0_true} against $H_1:\beta>\beta_0$. The test of \citet{Hemerik.etal:2020} rejects the null hypothesis if 
\begin{equation}\label{example.test}
    S_{1}>S_{(\lceil (1-\alpha)g \rceil)}
\end{equation}
where $\lceil \cdot \rceil$ represents the ceiling function. Analogous procedures for $H_1:\beta<\beta_0$ or $\beta \ne \beta_0$ are straightforward.

The definition of the test involves $W$, which is unknown. In practice we only have an estimate of $W$ available, which converges to $\asW$ by Assumption 1. For the theoretical results of the remainder of this paper, we will treat $\asW$, though not $W$, as known. To motivate this, we note  that any error terms relating to estimation of $\asW$ are of lower asymptotic order with respect to the parameter estimation \citep{Barndorff:1989, pace1997principles}. 



In the special case of the linear model, the effective sign-flip score test \eqref{eq:eff} simplifies to $S(\obsFlip)=\sigma n^{-1/2}X^T(I-H)V^{-1/2}\obsFlip (y-{\estmu})=\sigma n^{-1/2}X^T(I-H)\obsFlip (I-H)y$ (where \eqref{eq:hproj} reduces to $H=Z(Z^T Z)^{-1}Z^T$).
In that case, it relates closely to methods that resample the residuals, summarized in the excellent reviews of \cite{Winkler.et.al.2014} and \cite{permuco}.
\citet{kennedy1996randomization}, \citet{freedman1983nonstochastic} and \citet{dekker2007}
 use the same test statistic, while that of the Still-White procedure and of \citet{draper1966testing} is the so-called ``basic score'' of \citet{Hemerik.etal:2020}: $\sigma n^{-1/2}X^T\obsFlip (I-H)y$. \citet{ter1992permutation}, in contrast, makes permutations of the residuals estimated under the full model. 
In all these methods, the null distribution is derived by permuting (or sign-flipping) the residuals and, subsequently refitting the linear model.
This re-fitting approach can not be extended to GLMs, since the pseudo-responses does not retain crucial characteristics of the original response (e.g., count data may not be integer after permuting the residuals, or binomial responses are not between 0 and 1 anymore). This invalidates the required refitting step of the methods, which would also be computationally costly in GLMs.
In contrast, \cite{Hemerik.etal:2020} prove that their effective score statistic \eqref{eq:eff} has asymptotically the same distribution (under the null hypothesis) both for observed and flipped test statistics, hence avoiding the refitting and permitting easy extension to the GLM context.

\section{Standardized sign-flip score test}\label{sect:standardized}

We restate the result of \cite{Hemerik.etal:2020} concerning the validity of their test below for the special case of GLMs. Our alternative proof, given in the Appendix with all other proofs, emphasizes the importance of asymptotic null-invariance.

\begin{thm} \label{theo:asympt_eff}
Assume that the variances are correctly specified, that is, $\asV=V$, and that Assumption \ref{ass:link}-\ref{ass:lindeberg} hold.  For $n\xrightarrow{} \infty$, the test that rejects $H_0$ if \eqref{example.test} holds is an asymptotically $\alpha$ level test.
\end{thm}

If null-invariance holds only asymptotically, the distributions of $S(\II)$ and $S(\flip)$ are not identical, but both converge to the same distribution. It is natural to suppose that the more aspects of the distributions of $S(\II)$ and $S(\flip)$ are equal to each other in all finite samples, the closer the resulting test is to an exact test. 

The test of \cite{Hemerik.etal:2020} does not have the second-moment null-invariance property (Section \ref{Subsec:sec_mom_null_inv}): though the first moments of $S(\II)$ and $S(\flip)$ are equal in finite samples, the variances are not. In fact, we prove in Theorem 2 below that for finite samples $S(\II)$ always has a larger variance than $S(\flip)$, at least in the linear model with correctly specified variance. As a  consequence, the test shows a tendency to anti-conservativeness in finite samples. Since the variance of the observed test statistic is always larger than its flipped counterpart, extreme values in $S(\II)$ are more probable than in the reference distribution of $S(\flip)$. Consequently, the test tends to reject the null hypothesis too often.


\begin{prop}\label{prop:anticons}
Consider a normal regression model with identity link. Assume that the variances are correctly specified, that is, $\asV=V$, and that Assumption \ref{ass:link}-\ref{ass:lindeberg} hold. For finite sample size, the effective sign-flip score statistic defined as in \eqref{eq:eff} has $\Var\{S(\II)\} > \Var\{S(\flip)\}$.
\end{prop}

Proposition \ref{prop:anticons} is formulated for the linear model only. In the non-linear and/or variance-misspecified case $\Var\{S(\II)\}$ and $\Var\{S(\flip)\}$ are also unequal in general. In that case the term in the asymptotic expansion that makes $\Var\{S(\II)\} > \Var\{S(\flip)\}$ in the linear model is also present. However, in that case it is not the only asymptotic term, so it is difficult to make a fully general statement on anti-conservativeness. However, also in such models we see a tendency to small-sample anti-conservativeness in all simulations, as was also observed by \cite{Hemerik.etal:2020}.

The concept of second-moment null-invariance suggests that we can improve level accuracy of the test if we can modify the flipped scores to have equal variances.
The following result allows for such a procedure. It provides an expression for the variance of the flipped score, depending on the sign flip that has been applied.

\begin{lem}\label{lem:variancecorrection}
The variance of the sign-flipped score, as depending on $F$, is
\begin{equation}\label{eq_variance}
\Var\{S(\obsFlip)\} = n^{-1}X^T W^{1/2} (I-H)\obsFlip(I-H)\obsFlip (I-H)W^{1/2}X+o_p(1).
\end{equation}
\end{lem}
These variances can be estimated by plugging in $\hat{\gamma}$ and $\tilde W$. By dividing the flipped scores by their standard deviations, we 
obtain what we call the standardized sign-flip score statistics,
\begin{equation}\label{eq_standardized}
S^*(\obsFlip) = S(\obsFlip)/\Var\{S(\obsFlip)\}^{1/2}.
\end{equation}
We use the statistics $S^*(\obsFlip)$ in the same way as the original test uses the statistics $S(\obsFlip)$.
The estimate of $\Var\{S(\obsFlip)\}$ can be calculated for each $F$ in linear time in $n$, as we show in Lemma \ref{lem:computational} in the Appendix. 

Analogously with the test defined in \eqref{example.test}, consider without loss of generality testing \eqref{H_0_true} against $H_1:\beta>\beta_0$. Take the test that rejects the null hypothesis if
\begin{equation}\label{example.teststd}
    S^*_{1}>S^*_{(\lceil (1-\alpha)g \rceil)}.
\end{equation}
The following proposition establishes the null invariance and validity of the test.
\begin{prop}\label{prop:standard}
Assume that the variances are correctly specified, that is, $\asV=V$, and that Assumptions \ref{ass:link}-\ref{ass:lindeberg} hold. The standardized sign-flip score statistic is finite sample second-moment null-invariant. The test is asymptotically exact.
\end{prop}

The proposed standardization conditional on the sign flip matrix $F$ has a similar purpose as the refitting done by methods designed for the linear model, as described in Section \ref{Subsec:sec_mom_null_inv}. It addresses both issues in that standardization than make such methods difficult to generalize to GLMs. First, there is no refitting to incongruent pseudo-observations; second, the computational effort of standardization is limited.

\section{Robustness of the test} \label{sec:robust}

\cite{Hemerik.etal:2020} showed promising simulation results suggesting robustness against variance misspecification, but a formal proof was given only for the special case that all variances were misspecified by the same multiplicative constant. In this section we provide a formal proof of general robustness of the test to variance misspecification, both for the novel test derived in the previous section and for the original effective score test of \cite{Hemerik.etal:2020}. This is the most important result of this paper.

Incorrect specification of the model variance means that
\begin{equation*}
    \E\left\{(y-\mu)(y-\mu)^T\right\}=V \ne \asV,
\end{equation*}
where, as we recall $V$ is the true covariance of the outcomes, and $\asV$ is the limit of the estimates under the chosen estimation procedure. We first revisit the simple case considered by \cite{Hemerik.etal:2020}, in which the variance is misspecified by a multiplicative constant. In this case, the properties of the standardized sign-flip score test are not affected by this misspecification. In particular, we have second-moment null-invariance. 

\begin{prop}\label{prop:constrobust}
Assume that Assumptions \ref{ass:link}-\ref{ass:lindeberg} hold. If the variances are misspecified by any finite constant $c>0$, that is, $V=c\asV$, the standardized sign-flip score statistic is second-moment null-invariant. The test that rejects $H_0$ if \eqref{example.teststd} holds is  asymptotically exact.
\end{prop}

The proposition uses the property that the test is invariant to multiplication by a constant. This result is relevant, for instance, in models with a common unknown dispersion parameter, such as normal regression model. In such cases the test can be performed, and retains second-moment null-invariance, without the need to estimate the common dispersion parameter. We can save ourselves the effort of estimating it, taking simply $\estphi=1$. This special situation coincides with the standard parametric framework based on the quasi-likelihood approach. 

The main result of this paper concerns the case of general variance misspecification. The following theorem, which is a strong improvement over the properties shown by \cite{Hemerik.etal:2020}, shows that we can still get an asymptotic exact test.
\begin{thm}\label{theo:robust}
Assume that the variances are misspecified, that is, $V \ne \asV$ and that Assumptions \ref{ass:link}-\ref{ass:lindeberg} hold. For $n\xrightarrow{} \infty$, the standardized sign-flip score statistic is asymptotically second-moment null-invariant. The test that rejects $H_0$ if \eqref{example.teststd} holds is an asymptotically $\alpha$ level test.
\end{thm}

The following corollary enlarges the robustness property to the effective sign-flip score test, as an immediate consequence from the proof of the preceding theorem.
\begin{cor}
Assume that the variances are misspecified, that is, $V \ne \asV$, and that Assumptions \ref{ass:link}-\ref{ass:lindeberg} hold. For $n\xrightarrow{} \infty$, the effective sign-flip score statistic is asymptotically second-moment null-invariant. The test that rejects $H_0$ if \eqref{example.test} holds is an asymptotically $\alpha$ level test.
\end{cor}

\section{Multivariate test}\label{sect:multivariate}
Until now we have considered hypotheses about a single parameter $\beta \in \mathbb{R}$. We now generalize the test of the previous section to $\beta \in \mathbb{R}^d, \; d < n-q$. We consider a standard asymptotic setting where $d$ is fixed while $n$ increases. The null hypothesis of interest is now given by:

\[\label{H_0_multiv}
H_0: \beta=\beta_0 \in \mathbb{R}^d |\ (\gamma, \phi_1, \dots, \phi_n) \in \Gamma \times \Phi \times \dots \times \Phi, 
\]
where $\Gamma \subseteq \mathbb{R}^q$ and $\Phi \subseteq (0,\infty)$, which reduces to the null hypothesis \eqref{H_0_true} if $d=1$.

For this multivariate setting, we have to generalize the assumptions of Section \ref{section:glm}. Assumption \ref{ass:dispersion} remains unchanged while Assumptions \ref{ass:fisherinfo} and \ref{ass:lindeberg} are replaced by their multivariate counterparts, respectively,
\begin{ass}\label{ass:fisherinfomult}
The $\lim_{n \xrightarrow{} \infty} n^{-1} \tilde{\mathcal{I}}_{\beta,\beta}$ converges to a positive definite matrix.
\end{ass}

\begin{ass}  \label{ass:lindebergmult}
We require, for all $\epsilon >0$,
\begin{equation*}
    \lim_{n \xrightarrow{} \infty} \frac{1}{n} \sum_{i=1}^n E[\norm\tilde{\nu}_{i,\beta}^{*}\norm^2 \mathbf{1}_{\{\norm\tilde{\nu}_{i,\beta}^*\norm/\sqrt{n}>\epsilon\}}] \xrightarrow{} \mathbf{0},
\end{equation*}
and $n^{-1}\sum_{i=1}^n \Var(\tilde{\nu}_{i,\beta}^{*})$ to converge to a positive definite matrix, where $\norm \cdot \norm$ denotes the $\ell^2$ norm and $\mathbf{0}$ is a $d$-dimensional zero vector.
\end{ass}

\citet[Section 4]{Hemerik.etal:2020} derived a generalization of the sign-flipped effective score test statistic, as follows. Noting that the effective score is now a $d$-dimensional vector $S(F)=(S^1(\obsFlip),\dots,S^d(\obsFlip))^T$, an asymptotically exact $\alpha$-level test can be constructed by using the idea of the nonparametric combination methodology \citep{pesarin2001}. The test statistic takes the form 
\begin{equation*}
    T(F)=\left\{S(F)\right\}^T M \left\{S(F)\right\}
\end{equation*}
where $M$ is any non-zero matrix. Usually $M$ is chosen to be a symmetric matrix, and in general this choice influences the distribution of the power between the alternatives (see \cite{Hemerik.etal:2020} for details). A common choice for $M$ is the inverse of an estimate of the effective Fisher information of $\beta$, if available.

As in Section \ref{sect:standardized}, we can improve control of type I error by standardizing the score vector. Indeed, the same reasoning of Theorem \ref{prop:standard} applies, showing that the sign-flip standardized score vector is finite sample second-moment null-invariant. The definition of the test is analogous by noting that now $$\Var\{S(\obsFlip)\} = n^{-1}X^T W^{1/2} (I-H)\obsFlip(I-H)\obsFlip (I-H)W^{1/2}X+o_p(1)$$ is a $d \times d$ matrix and hence the standardized score 
\begin{equation*}
    S^*(\obsFlip) = S(\obsFlip)/\Var\{S(\obsFlip)\}^{1/2}
\end{equation*}
is a $d$-dimensional vector. We can therefore define the test statistic as
\begin{equation*}
    T^{*}(F)=\left\{S^*(F)\right\}^T  M  \left\{S^*(F)\right\}.
\end{equation*}

Assume we observe values $T^*_1,\dots, T^*_g$, where $T^*_1=T^*(\II)$ is the observed test statistic, while the sorted values are $T^*_{(1)}\le\dots\le T^*_{(g)}$.
Consider the test that rejects the null hypothesis if
\begin{equation}\label{example.mult}
    T^*_{1}>T^{*}_{(\lceil (1-\alpha)g \rceil)}.
\end{equation}
The following Proposition states that the test is second-moment null-invariant and asymptotically exact.
\begin{prop}\label{prop:mult_standard}
Assume that the variances are correctly specified, that is, $\asV=V$, and that Assumptions \ref{ass:link}, \ref{ass:dispersion}, \ref{ass:fisherinfomult} and \ref{ass:lindebergmult} hold. The $d$-dimensional standardized sign-flip score vector is finite sample second-moment null-invariant. The test that rejects $H_0$ if \eqref{example.mult} holds is an asymptotically $\alpha$ level test.
\end{prop}

Finally, the following theorem shows that the robustness properties of Theorem \ref{theo:robust} are inherited by this multivariate extension.
\begin{thm}\label{theo:mult_robust}
Assume that the variances are misspecified, that is, $V \ne \asV$ and that Assumptions \ref{ass:link}, \ref{ass:dispersion}, \ref{ass:fisherinfomult}, \ref{ass:lindebergmult} hold. For $n\xrightarrow{} \infty$, the $d$-dimensional standardized sign-flip score vector is asymptotically second-moment null-invariant. The test that rejects $H_0$ if \eqref{example.mult} holds is an asymptotically $\alpha$ level test.
\end{thm}

\section{Simulation study}\label{sect:simulation}
\begin{figure}[!ht]
\includegraphics[width=12cm]{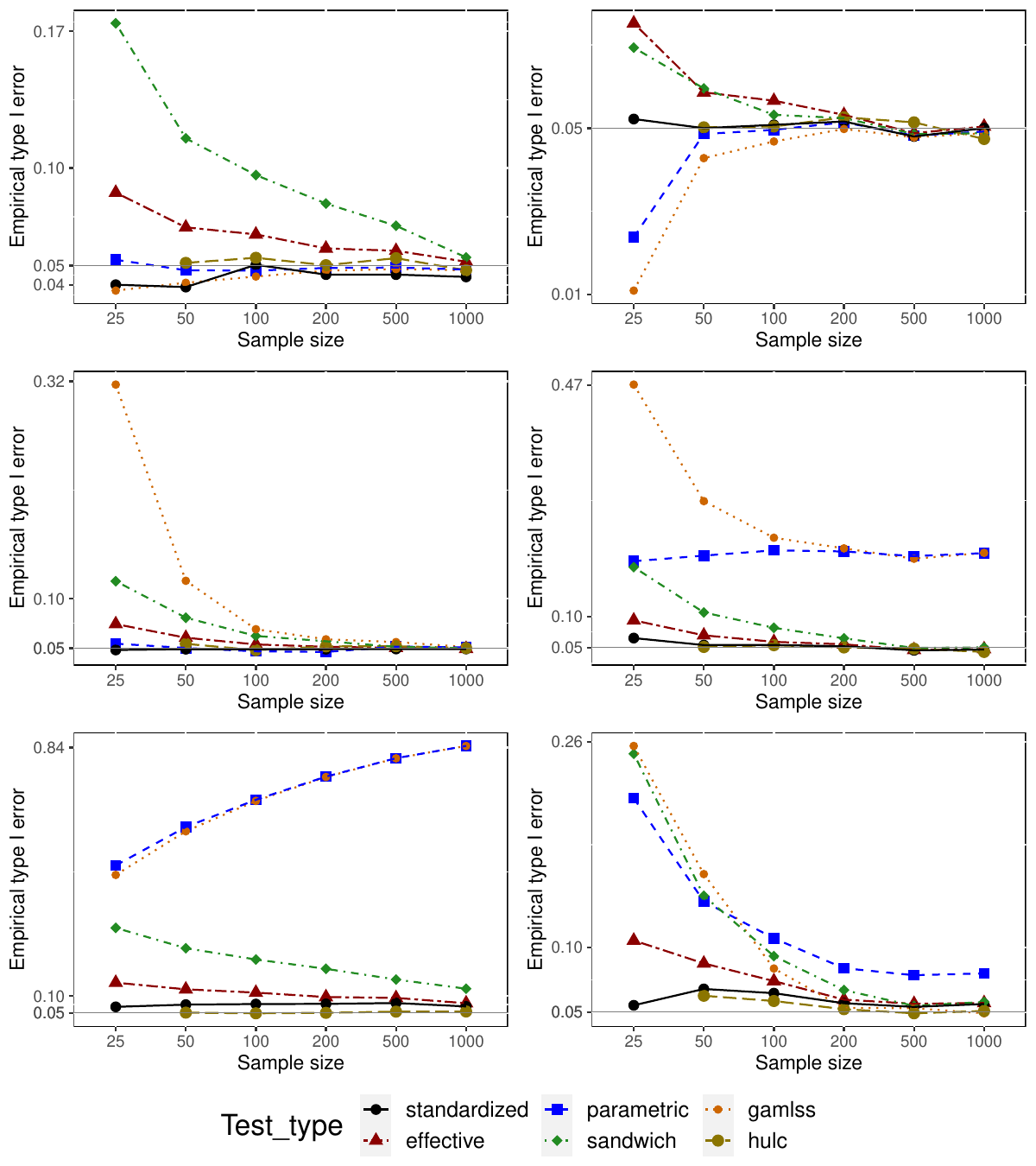}
\caption{Type I error control comparison. \textbf{Top-left:} correct Poisson model. \textbf{Top-right:} correct Logistic model. \textbf{Middle-left:} Normal with nuisance heteroscedasticity. \textbf{Middle-right:} Normal with target heteroscedasticity. \textbf{Bottom-left}: false Poisson model. \textbf{Bottom-right:} two groups Negative-binomial.}
\label{fig:errorcontrol}
\end{figure}

\begin{figure}[!ht]
\includegraphics[width=12cm]{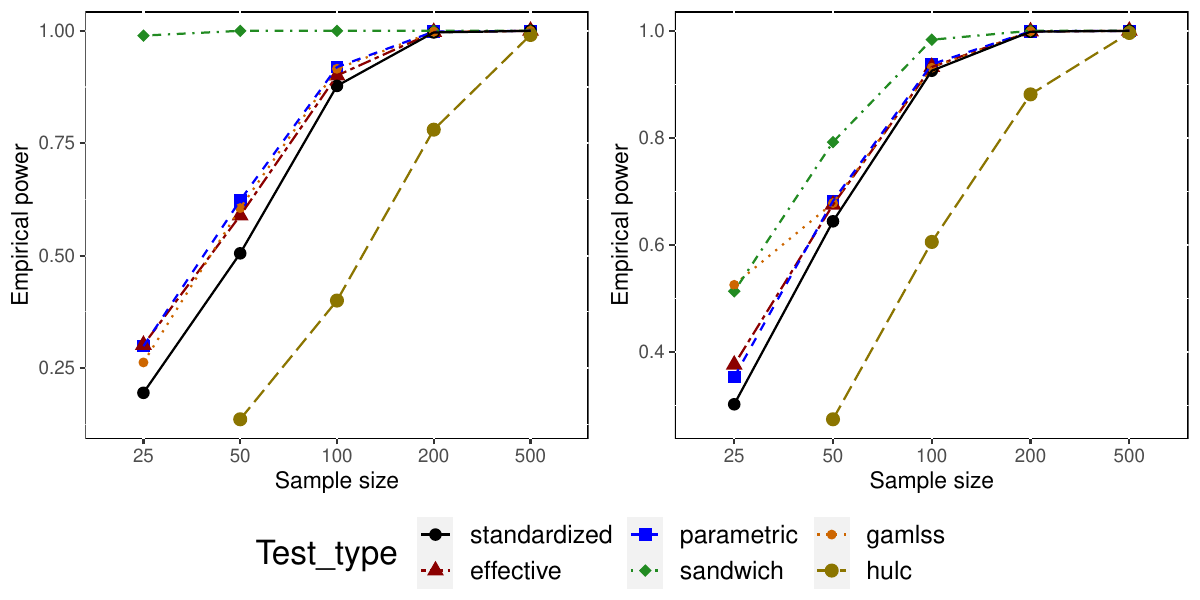}
\caption{Power comparison. \textbf{Left:} correct Poisson model. \textbf{Right:} correct Normal model. }
\label{fig:powercontrol}
\end{figure}
We explore six different settings to compare empirically the type I error control of the usual parametric approach (by considering the Wald test), the effective and standardized flip-score tests, the Wald test based on the use of the sandwich estimator of the variance to correct for variance misspecification, the GAMLSS, modeling the variance as a function of the full model, and the HulC method, inverting the estimated confidence interval to perform hypothesis testing. A total of 5000 simulations have been carried out for each setting. The covariates have been drawn from a multivariate normal distribution, with $X \in \mathbb{R}$ and $Z \in \mathbb{R}^3$. The three nuisance covariates have correlation with the target variable equal to $(0.5,0.1,0.1)$, while the true parameter is set to $\beta=0$. The null hypothesis considered is $H_0:\beta=0$ against a two-sided alternative, with a nominal significance level $\alpha=0.05$. Different sample sizes are considered $(25,50,100,200,500,1000)$. The sign-flip tests are performed through the {\tt R} library {\tt flipscores}. Note that for the HulC method we could not perform simulations for the smallest sample size, since after the sample splitting there were not enough observations (only four) to evaluate the model in the sub-samples \riccardo{in the setting considered}.

The top two plots of Figure \ref{fig:errorcontrol} are settings with correctly specified models, respectively a Poisson and a Logistic regression models. The middle two plots represent normal models with neglected heteroscedasticity, which depends either on a nuisance covariate or on the tested variable, i.e. respectively $\Var(y_i)=4z^2_i$ and $\Var(y_i)=4x^2_i$. In the bottom-left plot of Figure \ref{fig:errorcontrol} a Poisson model was fitted when the true distribution was negative binomial with dispersion parameter $\phi=1$, that is, additional heteroscedasticity relative to the Poisson model that depends on the mean. The bottom-right plot displays the results of a two-sample test fitting a negative binomial. A common dispersion parameter was assumed for the fitted  negative binomial model, but the two groups are unbalanced (proportion equal to 2/3 and 1/3), and generated from two distributions with different dispersion parameter ($0.4$ and 1).

In the two top plots we observe that, in case of correctly specified models, the standardized test is close to the nominal level, even with $n=25$, improving the asymptotic convergence of the effective test. The parametric test shows few rejections with small sample size for the logistic model, due to a poor approximation of the likelihood. The sandwich shows slow convergence in both cases, which is unsatisfactory since we are dealing with a well-specified model. The GAMLSS shows a low convergence, explained by the greater number of parameters involved. The HulC method is always appreciably close to the nominal level. The middle-left plot shows a similar behavior. 

The last three plots show the failure of the parametric test in presence of some forms of variance misspecification, where the rejection fraction converges to a level far from the nominal for increasing sample size. The standardized test outperforms its competitors in all cases, being closest to the nominal level. In particular, the improvement over the effective test and further over the sandwich test is clear. Further, the GAMLSS surprisingly shows a failure in the middle-right plot, while in the bottom-left the reason is due to the impossibility of modeling the variance adopting a Poisson model. In the last plot it shows a slower convergence compared to the sign-flip tests. On the other side the HulC method is always close to the nominal level, with the best performance in the two bottom plots.

In the final two plots it is remarkable to see that the standardized test seems to converge to the effective test faster than going to the nominal level, leading to an initial worsening of the true level, with a recovery for larger sample size.

Finally, Figure \ref{fig:powercontrol} contains an evaluation of the power of the tests with two well-specified models, respectively a Poisson and a Normal model, with true parameter equal respectively to 0.3 and 1. The results for the sandwich estimator are given for completeness, although they are not comparable for small sample size, since that method has no control of the type I error. We see that the improvement of type I error control of the standardized test with respect to the effective naturally costs some power. Analogously, we see some power loss also with respect to the parametric model and the GAMLSS as expected, but this difference is remarkably small. On the other side, the loss in power of the HulC method is remarkable, much higher with respect to the standardized test.

\section{Real data analysis: RNA-Seq data}\label{section:realdata}
In RNA sequencing data \citep{RnaSeq} a common aim is to find genes that are differentially expressed across a group of units. The usual analysis adopts a negative binomial regression model, since the observed target variables are counts, and overdispersion relative to the Poisson distribution is standard. However, the variance model generally assumes a fixed mean-variance structure with a common dispersion parameter among the groups of interest, which can be problematic, as we will show. 

From the Cancer Genome Atlas (TCGA) \citep{TCGAdataset}, we have taken the TCGA-LIHC dataset of Liver Hepatocellular Carcinoma (HCC) \citep{TCGALIHCdataset}. The TCGA-LIHC consists of 20\,119 genes for 344 patients with a primary tumor. We performed a very limited pre-filtering, deleting only the genes with zero total count. The target covariate is the pathological stage of the tumor. We treat it as a binary variable, splitting it between first pathological stage versus all higher stages. A total of 170 patients have a first pathological stage of the tumor, while 174 patients have higher stages. We further included two covariates in the fitted model: gender and age. \riccardo{Some summary statistics are detailed in the supplementary material.}

The state-of-the-art method for the analysis of these data, DESeq2 \citep{RnaSeq}, uses a negative binomial model with a dispersion parameter that is allowed to differ between genes, but does not depend on covariates. In particular, it does not depend on pathological stage. We tested this assumption for each gene using a GAMLSS model. The null hypothesis that dispersion did not depend on tumor stage was rejected for 5\,967 out of 20\,119 genes at the unadjusted 5\% level. Since we would expect approximately 1\,000 rejections if the DESeq2 model fits well, this gives clear indication of lack of fit of that model, at least in some of the genes. Based on the simulations in Section \ref{sect:simulation} we would, therefore, expect DESeq2 to be anti-conservative for these data.

Next, we fitted both the Poisson regression model and the negative binomial regression to each gene, and applied the Standardized sign-flip test using $5\,000$ flips (the default choice). We adjusted for multiple testing using the Benjamini-Hochberg correction \citep{BenjHoch95} at $\alpha=0.05$.
We compared the results with those of DESeq2 \citep{RnaSeq}.
This procedure estimates the nuisance parameters with an approach that shares information across genes. After obtaining a p-value for each gene, DESeq2 also applies the Benjamini-Hochberg correction at $\alpha=0.05$ internally.

DESeq2 outputs a total of $1\,450$ NA-values, due to its automatic pre-filtering: $1\,059$ genes were filtered out due to detected outliers and 391 because of low counts. This pre-filtering is meant to increase the power of the method, since it reduces the multiple testing burden, while removing genes for which the method has low power anyway. In contrast, the Negative Binomial regression gave 88 NA-values due to lack of convergence, while the Poisson regression gives no errors. 

Table \ref{tab:my_label} gives the raw number of rejections for DESeq2 and the sign-flip tests in the second column, while the third column shows the number of rejections considering only the genes where no methods returned NA-values. 

\begin{table}[h]
        \centering
        \begin{tabular}{||c|c|c||}
            \hline Method & No. of rejections & Filtered No. of rejections \\ \hline DESeq2&$3\,360$ & $3\,358$ \\Poisson sign-flip&$5\,109$ & $5\,032$\\Negative binomial sign-flip&$4\,833$ & $4\,765$\\ \hline
        \end{tabular}
        \caption{Number of rejections for three methods}
        \label{tab:my_label}
    \end{table}

We found that for this dataset, the standardized sign-flip test for the Poisson and Negative Binomial regression perform similarly, in fact they share the 96.1\% of the conclusions. This confirms that our test is not so sensitive to (wrong) assumptions about the variance function. In contrast, DESeq2 obtains fewer rejections than both sign flip methods, despite the risk of anti-conservativeness due to the misspecification of the model.

\begin{figure}[!ht]
\includegraphics[width=12cm]{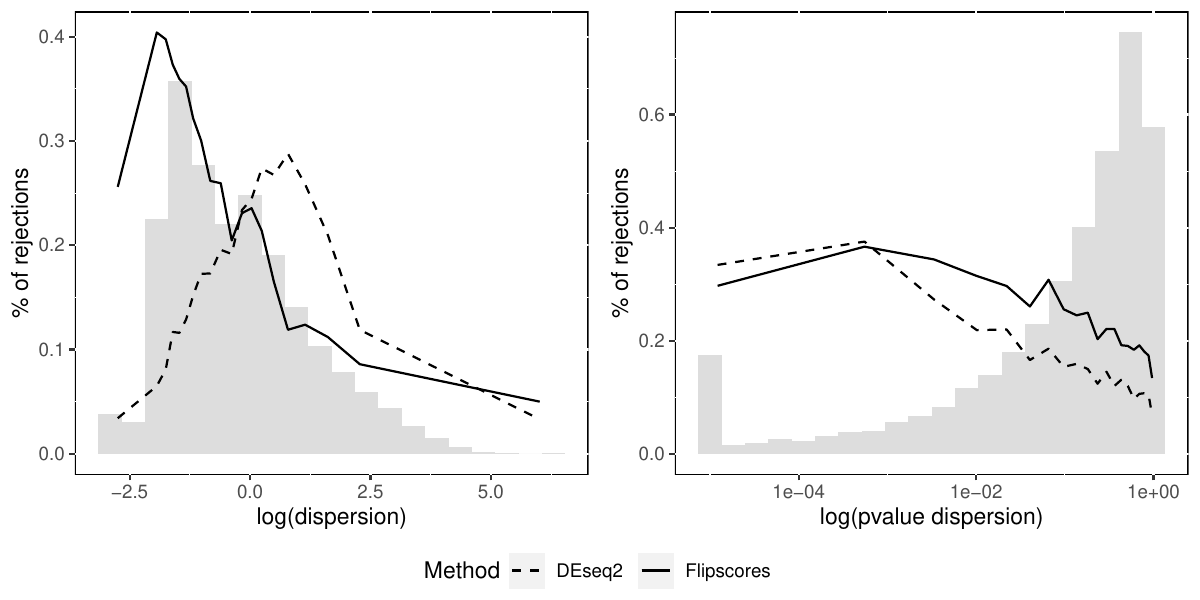}
\caption{Percentage of rejections \textbf{Left:} compared to the dispersion parameter. \textbf{Right:} compared to the p-value of the test for equality of dispersion between groups.}
\label{fig:rnaseqplot}
\end{figure}

Figure \ref{fig:rnaseqplot} shows the proportion of rejections as function of the maximum likelihood estimates of the dispersion parameters, on the left, and as function of the $p$-value of the test of equality of the dispersion parameter between the two groups. Further, the overlaid histogram represents the frequencies of the values on the $x$ axis. Since the two sign-flip tests have very similar results, we plot only the Negative Binomial sign-flip test against DESeq2. In the left-hand panel we observe higher power of the sign-flip test when the dispersion parameter is low, which is true for most genes. In the right-hand panel, we see a high number of rejections of DESeq2 when there is evidence for lack of equality of the dispersion parameter between groups, which can be explained from its lack of type I error control when this key assumption does not hold. When there is no such evidence, we observe higher power for the sign-flip test.


\section{Discussion}\label{sect:discussion}
Variance misspecification can dramatically reduce the quality of inference in any chosen model. We have seen that it is especially true for type I error control in hypothesis testing. In this paper we have developed a method which can be applied in the broad class of generalized linear models, where it is often difficult to properly check the assumption about the variance structure. Proper variance modeling in generalized linear models is crucial in important application areas such as RNA sequencing.

We have derived a novel sign-flipping test with two important properties. First, if the model is correctly specified, the test is marginally second-moment null-invariant. As a consequence, it converges to the nominal level extremely fast, with comparable or sometimes even better control for small sample size than the parametric test. Second, if the variance model is misspecified in any arbitrary way, the new test is still asymptotically correct under minimal assumptions. Simulations show a huge improvement in small sample performance over tests based on the sandwich estimator or GAMLSS models. On the other side, the HulC method is comparable for the type I error control in misspecified problems, but only our method is applicable for the smallest sample sizes and, further, the consequences of sample splitting are visible, in the sense that we observe a relevant loss of power for HulC compared to our proposed method. 


We have emphasized the value of marginal second-moment null-invariance, even though this does not imply full marginal invariance and consequently does not result in finite sample exactness. However, we note that in GLMs, unlike in the linear model,  parametric tests are also unable to achieve exact control of the type I error.

\riccardo{
Finally, we note that by inverting our test, an asymptotically exact confidence interval for the tested coefficient can be obtained. For very simple permutation-type tests, such a confidence interval can be obtained in a fast and exact way \citep{tritchler1984inverting}. For nonparametric tests for linear and generalized linear models, a numerical method can be used to invert the test with any desired accuracy. 
}

\bibliographystyle{plainnat}
\bibliography{bibtex.bib}
\newpage


\appendix
\section{Appendix: Proofs of the Theorems}
We first introduce some Lemmas needed for the proofs of the Theorems.

The projection matrix for GLMs introduced in \eqref{eq:hproj} is a proper projection matrix for studentized units as shown in the following Lemma, which is mentioned for instance in \cite[p. 136]{Agresti:2015}. We give a full proof for the sake of completeness, since that is omitted in most textbooks.
\begin{lem}\label{lem:hmat}
    The fitted and observed values in a GLM are connected through the relation
    \begin{equation*}
    V^{-1/2}(\hat{\mu}-\mu)=HV^{-1/2}(y-\mu)\{1+o_{p}(1)\}    
    \end{equation*}
\end{lem}
\begin{proof}
The first-order approximation of the score function is
\begin{align*}
    s_{\hat{\beta}}&=s_{\beta}+\mathcal{I}_{\beta,\beta}(\hat{\beta}-\beta)+o_p(1)\\
    0&=X^TD{V}^{-1}(y-\mu)-X^T{W}X(\hat{\beta}-\beta)+o_p(1)\\
    \hat{\beta}-\beta&=(X^T{W}X)^{-1} X^TD{V}^{-1}(y-\mu)+o_p(1)
\end{align*}
where $(o_{p}(1))$ is an error term asymptotically negligible.
Then, by the Delta method we have 
\begin{align*}
    \hat{\mu}&=\mu+\frac{d\mu}{d\eta}\frac{d\eta}{d\beta}(\hat{\beta}-\beta)\{1+o_{p}(1)\}\nonumber\\
    \hat{\mu}-\mu&=DX^T(\hat{\beta}-\beta)\{1+o_{p}(1)\}\nonumber\\
    &=DX^T(X^T{W}X)^{-1} X^TD{V}^{-1}(y-\mu)\{1+o_{p}(1)\}\nonumber\\
    &=D{W}^{-1/2}{H}{W}^{1/2}D^{-1}(y-\mu)\{1+o_{p}(1)\}\nonumber\\
    &={V}^{1/2}{H}{V}^{-1/2}(y-\mu)\{1+o_{p}(1)\}\nonumber \\
    {V}^{-1/2}(\hat{\mu}-\mu)&={H}{V}^{-1/2}(y-\mu)\{1+o_{p}(1)\}.  
\end{align*}
where the asymptotically negligible error term has two sources, one related to the second-order approximation of the likelihood and the other related to possible non-linearity of the link function.
\end{proof} 

\begin{lem}\label{lem:quadform}
    Let $C$ be any $n$-dimensional matrix and $G$ be a nonsingular $n\times n$ matrix, then $C$ and $G^{-1}CG$ have the same set of eigenvalues (with the same multiplicities).
\end{lem}
\begin{proof}
See \cite[p. 15]{Magnus2019}.
\end{proof}

\begin{lem}\label{lem:diag}
Let $C$ be any $n$-dimensional matrix and $\flipSpace$ be the set of all $n$-dimensional flipping matrices, i.e. the set of all possible $n$-dimensional diagonal matrices $\obsFlip$ with elements $-1$ or $1$. Then
\begin{equation*}
    \sum_{\obsFlip \in \flipSpace}
    \obsFlip C \obsFlip = 2^n \mbox{diag}(C),
\end{equation*}
where $\mbox{diag}(C)$ is a diagonal matrix with the same diagonal elements of $C$.
\end{lem}
\begin{proof}
First we note that the absolute value of each element of $C$ does not change after the multiplication $\obsFlip C \obsFlip$. The sums for the off-diagonal elements contain an equal number of terms with positive and negative sign, hence their sum over all possible flips is zero, while the sign of each diagonal element is positive for each term. By noting that the total number of flips is $2^n$ we have the claim.
\end{proof}

The following result was adapted from \cite{Huber2009}, extending their theorem (Proposition 7.1, p. 156) for linear regression models to GLMs.
\begin{lem}
\label{lem:hspar}
Assume that the regression coefficients of a generalized linear model are consistently estimated, in the sense that for every $\varepsilon>0$, as $n\rightarrow\infty$,
$$\max_{1\leq i \leq n} \mbox{pr}(|\hat{\mu}_i-\mu_i|>\varepsilon)\rightarrow 0.$$
Then
\begin{equation*}
\max_{1\le i\le n} h_{ii} \xrightarrow{} 0,
\end{equation*}
where $h_{ik}$ is the $ik$-th element of the matrix $H$, as defined in \eqref{eq:hproj}.
\end{lem}

\begin{proof}
Using Lemma \ref{lem:hmat} we have, for each $i$,
$\hat{\mu}_i-\mu_i=$
\begin{equation} \label{equ}
h_{ii} (y_i-\mu_i)+\sum_{k\ne i} v_i^{1/2} v_k^{-1/2} h_{ik} (y_k-\mu_k)+o_p(1).
\end{equation} 
Now we give a general probability result. Let $V_1$ and $V_2$ be two independent random variables, for any $\varepsilon>0$ we have that
\begin{align*}
\mbox{pr}(|V_1+V_2|&\ge \varepsilon)\ge \mbox{pr}(V_1\ge \varepsilon)\mbox{pr}(V_2\ge 0)+ \mbox{pr}(V_1\le -\varepsilon)\mbox{pr}(V_2\le 0)\\&\ge \min \left\{\mbox{pr}(V_1 \ge \varepsilon), \mbox{pr}(V_1 \le -\varepsilon) \right\}.
\end{align*}
Noting that $(y_i-\mu_i)$ is independent from $(y_k-\mu_k)$ for each $k\ne i$ we can apply the result to expression \eqref{equ} to obtain
\begin{equation*}
\mbox{pr}\left( |\hat{\mu}_i-\mu_i \ge \varepsilon|\right)\ge \min \left[\mbox{pr}\left\{(y_i-\mu_i) \ge \frac{\varepsilon}{h_{ii}}\right\},\mbox{pr}\left\{(y_i-\mu_i) \le -\frac{\varepsilon}{h_{ii}}\right\}\right]=m_i^n. 
\end{equation*}
We have $\max_{1\leq i \leq n} \mbox{pr}(|\hat{\mu}_i-\mu_i|>\varepsilon)\rightarrow 0$, so $\max_{1\le i\le n} m_i^n \xrightarrow{} 0$.
Since $\varepsilon$ was arbitrary, this implies that $\max_{1\le i\le n} h_{ii} \xrightarrow{} 0$.
\end{proof}

\begin{lem} \label{lem:computational}
The computational cost of the standardization constant in (\ref{eq_standardized}) is linear in $n$.
\end{lem}
\begin{proof}
Define $W^{1/2}Z=U\Delta L^T$, that is, the singular value decomposition of $
W^{1/2}Z$, where $U$ is a semiorthogonal $n\times q$ matrix ($q$ equal to the rank of $Z$), $\Delta$ a diagonal $q$ matrix and $L$ a $q\times q$ orthogonal matrix. Therefore \eqref{eq:hproj} can be written as

\begin{equation*}
H=W^{1/2}Z(Z^TW Z)^{-1}Z^TW^{1/2}=U\Delta L^TL\Delta U^T=UU^T.
\end{equation*}

Now, let $a=(I-H)W^{1/2}X$ and $A=diag(a)$. Further, let $\ones$ be an $n$-dimensional vector of ones.
The denominator of the standardized test statistic becomes
\begin{eqnarray*}
&&X^TW^{1/2}(I-H)\obsFlip(I-H)\obsFlip(I-H)W^{1/2}X =\\
=&&a^T\obsFlip(I-UU^T)\obsFlip a=\\
=&&a^T\obsFlip I\obsFlip a - a^T\obsFlip UU^T\obsFlip a\\
=&&a^Ta - \ones^TA\obsFlip UU^T\obsFlip A \ones\\
=&&a^Ta - \ones^T\obsFlip A U U^T A \obsFlip \ones\\
=&&a^Ta - f^T C C^T f,
\end{eqnarray*}
where $f=\obsFlip\ones$ and $C=AU$.

Therefore, given that $a^Ta$ is a constant and the computational cost of $f^T C C^T f$ is linear with $n$ since we can write $f^T C C^T f = \sum_{j=1}^q (\sum_{i=1}^n f_iC_{ij})^2$, the result of the lemma follows.
\end{proof}

\subsection{Proof of Theorem \ref{theo:asympt_eff}}
\begin{thm*}
Assume that the variances are correctly specified, that is, $\asV=V$, and that Assumption \ref{ass:dispersion}-\ref{ass:lindeberg} hold.  For $n\xrightarrow{} \infty$, the test that rejects $H_0$ if \eqref{example.test} holds is an asymptotically $\alpha$ level test.
\end{thm*}

\begin{proof}
By the definition of the random sign-flipping transformations is trivial to observe that the expected value of the test statistic $S(\flip)$ is zero.

By Lemma \ref{lem:hmat} we can rewrite the effective score statistic \eqref{eq:eff} as
\begin{equation*}
    S(\flip)=n^{-1/2}X^TW^{1/2}(I-H)\flip(I-H)V^{-1/2}(y-\mu)+o_p(1). 
\end{equation*}
Let 
\begin{equation}\label{def:avector}
    a=(I-H)W^{1/2}X.
\end{equation}
The variance conditional on $\flip=F$ is
\begin{align*}
    \Var\left\{S(\obsFlip)\right\}&=n^{-1}a^T\obsFlip(I-H)V^{-1/2}\E\{(y-\mu)(y-\mu)^T\}V^{-1/2}(I-H)\obsFlip a+o_p(1)\\
    &=n^{-1}a^T\obsFlip(I-H)V^{-1/2}V V^{-1/2}(I-H)\obsFlip a+o_p(1)\\
    &=n^{-1}a^T\obsFlip(I-H)\obsFlip a+o_p(1)
\end{align*}
and for $\flip=\II$ we have
\begin{equation*}
    \Var\left\{S(\II)\right\}=n^{-1}a^TIa+o_p(1).
\end{equation*}
Taking the difference
\begin{equation*}
\Var\left\{S(\II)\right\}-\Var\left\{S(\obsFlip)\right\}=n^{-1}a^T\left\{I-\obsFlip(I-H)\obsFlip\right\}a+o_p(1),
\end{equation*}
we note that the first term is a quadratic form and we look at the matrix $$I-\obsFlip(I-H)\obsFlip=I-\obsFlip\obsFlip+\obsFlip H\obsFlip=\obsFlip H\obsFlip.$$ 
Since $H$ is a projection matrix and $\obsFlip^{-1}=\obsFlip$, Lemma \ref{lem:quadform} implies that $\obsFlip H\obsFlip$ is positive semidefinite, so that asymptotically $\Var\left\{S(\II)\right\}-\Var\left\{S(\obsFlip)\right\}\ge 0$.

Define 
$$h_{\sup}=\sup_{1\le i \le n} \left(h_{ii}\right).$$
Let us make the randomness of the flips explicit. Let $\flipSpace $ denote the set of all possible flipping matrices, and note that $|\flipSpace|=2^n$. By Lemma \ref{lem:diag} we have
\begin{eqnarray*}
\E\left[\Var\left\{S(\II)\right\}-\Var\left\{S(\flip)\mid\obsFlip\right\}\right] &=& (2^n)^{-1}\sum_{\obsFlip \in \flipSpace}[ \Var\left\{S(\II)\right\}-\Var\left\{S(\obsFlip)\right\}] \\
&=& n^{-1}a^T\mathrm{diag}(H)a +o_{p}(1) \\
&\leq& n^{-1}\|a\|^2\cdot h_{\sup}
+o_{p}(1),
\end{eqnarray*}
where, by Assumption \ref{ass:fisherinfo} and Lemma \ref{lem:hspar}, the limiting behavior is
\begin{equation*}
    \lim_{n\xrightarrow{}\infty} n^{-1}\|a\|^2\cdot h_{\sup} = 0.
\end{equation*}
According to the law of total variance we have
\[
\Var\{S(\flip)\} = \Var\left[\E\left\{S(\flip)\mid\obsFlip\right\}\right] + \E\left[\Var\left\{S(\flip)\mid\obsFlip\right\}\right].
\]
We know that $\E\{S(\flip)\mid\obsFlip\}$ does not depend on $\flip$, which means that $\Var[E\{S(\flip)\mid\obsFlip\}]=0$, so $\Var\{S(\flip)\} = \E[\Var\{S(\flip)\mid\obsFlip\}].$ It follows that, marginally over $\obsFlip$,
\[ \lim_{n\xrightarrow{}\infty}
\Var\{S(I)\}-\Var\{S(\flip)\} = 0.
\]

Since the random sign-flipping transformations are all independent, the corresponding test statistics are all uncorrelated, i.e, for all $1\leq l <m\leq g$ we have $\mbox{cov}\{S(\flip_l),S(\flip_m)\}=0$.
By the previous results and Assumption \ref{ass:lindeberg},  $(S(\II),\dots,S(\flip_g))^T$ converges to a multivariate normal distribution by the multivariate Lindberg-Feller central limit theorem \citep{Vaart1998}, with mean vector $\textbf{0}$ and covariance matrix $ s^2\II,$ where $s^2$ is the limiting variance of the flipped test statistic.
Finally, we use Lemma 1 of \cite{Hemerik.etal:2020} to conclude that the test that rejects when \eqref{example.test} holds is an asymptotic $\alpha$ level test. 
\end{proof}

\subsection{Proof of Proposition \ref{prop:anticons}}
\begin{prop*}
Consider a normal regression model with identity link. Assume that the variances are correctly specified, that is, $\asV=V$, and that Assumption \ref{ass:dispersion}-\ref{ass:lindeberg} hold. For finite sample size, the effective score statistic defined as in \eqref{eq:eff} has $\Var\{S(\II)\} > \Var\{S(\flip)\}$.
\end{prop*}
\begin{proof}
By Lemma \ref{lem:hmat} we can rewrite the effective score statistic \eqref{eq:eff} as
\begin{equation*}
    S(\flip)=n^{-1/2}X^TW^{1/2}(I-H)\flip(I-H)V^{-1/2}(y-\mu)+o_p(1). 
\end{equation*}
The variance conditional on $\flip=F$ is
\begin{align*}
    \Var\left\{S(\obsFlip)\right\}&=n^{-1}a^T\obsFlip(I-H)V^{-1/2}\E\{(y-\mu)(y-\mu)^T\}V^{-1/2}(I-H)\obsFlip a\\
    &=n^{-1}a^T\obsFlip(I-H)V^{-1/2}V V^{-1/2}(I-H)\obsFlip a\\
    &=n^{-1}a^T\obsFlip(I-H)\obsFlip a\\
\end{align*}
and for $\flip=\II$ we have
\begin{equation*}
    \Var\left\{S(\II)\right\}=n^{-1}a^TIa.
\end{equation*}
Taking the difference
\begin{equation*}
\Var\left\{S(\II)\right\}-\Var\left\{S(\obsFlip)\right\}=n^{-1}a^T\left\{I-\obsFlip(I-H)\obsFlip\right\}a,
\end{equation*}
we note that it is a quadratic form and we consider the matrix $$I-\obsFlip(I-H)\obsFlip=I-\obsFlip\obsFlip+\obsFlip H\obsFlip=\obsFlip H\obsFlip.$$ 
Since $H$ is a projection matrix and $\obsFlip^{-1}=\obsFlip$, Lemma \ref{lem:quadform} implies that $\obsFlip H\obsFlip$ is positive semidefinite, and therefore $\Var\{S(\II)\}-\Var\{S(\obsFlip)\} \ge 0$.

We then prove the strict inequality for at least one flipping matrix $\obsFlip$. 
Taking any model with an intercept, by construction $$h_{\inf}=\inf_{1\le i \le n} \left(h_{ii}\right) > 0.$$
Note that $|\flipSpace|=2^n$.
Since $\Var\{S(\II)\}-\Var\{S(\obsFlip)\}\geq 0$ as proven above, it suffices to show that 
\[
\sum_{\obsFlip \in \flipSpace} n^{-1}a^T\obsFlip H\obsFlip a >0,
\]
which means that we have a strictly inequality for some $\obsFlip$.
Using Lemma \ref{lem:diag} we have
\[
\sum_{\obsFlip \in \flipSpace} \obsFlip H\obsFlip =2^n\,\mathrm{diag}(H) .
\]
Therefore
\begin{eqnarray*}
\sum_{F \in \flipSpace} n^{-1}a^T\obsFlip H\obsFlip a &=&2^n n^{-1} a^T\textrm{diag}(H)a \\
&\geq& 2^n n^{-1}\|a\|^2 \cdot h_{\inf} >0.
\end{eqnarray*}
\end{proof}

\subsection{Proof of Lemma \ref{lem:variancecorrection}}
\begin{lem*}
    The variance of the sign-flipped score, as depending on $F$, is
$$\Var\{S(\obsFlip)\} = n^{-1}X^T W^{1/2} (I-H)\obsFlip(I-H)\obsFlip (I-H)W^{1/2}X+o_p(1).$$
\end{lem*}

\begin{proof}
    Let $$a=(I-H)W^{1/2}X.$$
    The variance for a given sign-flix matrix $F$ is
\begin{align*}
    \Var\left\{S(\obsFlip)\right\}&=n^{-1}a^T\obsFlip(I-H)V^{-1/2}\E\{(y-\mu)(y-\mu)^T\}V^{-1/2}(I-H)\obsFlip a+o_p(1)\\
    &=n^{-1}a^T\obsFlip(I-H)V^{-1/2}V V^{-1/2}(I-H)\obsFlip a+o_p(1)\\
    &=n^{-1}a^T\obsFlip(I-H)\obsFlip a+o_p(1)
\end{align*}
\end{proof}

\subsection{Proof of Proposition \ref{prop:standard}}
\begin{prop*}
[copy from main text]]
\end{prop*}
\begin{proof}
We  observe that the expected value of the new statistic \eqref{eq_standardized} is left unchanged, while the standardization makes the variance of each flipped new statistic equal to 1. 
The same argument given in the last part of the proof of Theorem \ref{theo:asympt_eff} applies to deduce the asymptotic exactness of the test. 

In the special case of the normal model with the identity link, we have that $\Var(S(\obsFlip))$ does not depend on any unknown nuisance parameters. Moreover, in this model $S(\II)$ and $S(\flip)$ are normally distributed in finite samples.
Further, the $w$ test statistics are all uncorrelated due to the random sign-flipping of the underlying summands.
Hence the standardized sign-flip score statistic is second-moment exact.
\end{proof}

\subsection{Proof of Proposition \ref{prop:constrobust}}
\begin{prop*}
[copy from main text]]
\end{prop*}
\begin{proof}
We observe that
\begin{equation*}
    V\asV^{-1}=\phi\asphi^{-1}I=cI
\end{equation*}
and therefore
\begin{equation*}
    \Var\{S^*(\obsFlip)\}=c \quad \forall \obsFlip \in \flipSpace,
\end{equation*}
as was to be shown.
\end{proof}
\subsection{Proof of Theorem \ref{theo:robust}}
\begin{thm*}
Assume that the variances are misspecified, that is, $V \ne \asV$ and that Assumptions \ref{ass:dispersion}-\ref{ass:lindeberg} hold. For $n\xrightarrow{} \infty$, the standardized sign-flip score statistic is asymptotically second-moment null-invariant. The test that rejects $H_0$ if \eqref{example.teststd} holds is an asymptotically $\alpha$ level test.
\end{thm*}
\begin{proof}
It is trivial to see that the expected value of the test statistic is not affected by this misspecification.

Let 
\begin{equation*}
    B= V \asV^{-1},
\end{equation*}
note that it is a diagonal matrix and all its elements are finite and greater than zero by Assumption \ref{ass:dispersion}.
We compute again the variance of \eqref{eq:eff}, but now we consider the misspecification. 

Let $\asH$ and $\asA$ be the quantities defined in \eqref{eq:hproj} and  \eqref{def:avector} for $W=\asW$. The variance can be written as
\begin{align*}
    \Var\left\{S(F)\right\}&=n^{-1}\asA^TF(I-\asH)\asV^{-1/2}\E\left\{(y-\mu)(y-\mu)^T\right\}\asV^{-1/2}(I-\asH)F\asA+o_{p}(1)\\
    &=n^{-1}\asA^TF(I-\asH)B(I-\asH)F\asA+o_{p}(1).
\end{align*}
Take the difference
\begin{align*}
\Var\left(S(\II)\right)-\Var\left(S(F)\right)&=n^{-1}\asA^T\left\{B-F(I-\asH)B(I-\asH)F\right\}\asA+o_{p}(1)\\
&=n^{-1}\asA^T\left[F\{\asH B+(I-\asH)B\asH\}F\right]\asA+o_{p}(1).
\end{align*}
We notice that by Lemma \ref{lem:diag} 
\begin{equation*}
    \sum_{F \in \mathcal{F}} F(\asH B+B\asH-\asH B\asH)F = 2^n\,\mathrm{diag}(\asH B+B\asH-\asH B\asH)
\end{equation*}
where the $i$-th element of that diagonal matrix is 
\begin{equation*}
    2\ash_{ii}b_i-\sum_{k=1}^n \ash^2_{ik}b_k.
\end{equation*}
Then we have
\begin{align*}
\E \left[\Var\left\{S(\II)\right\}-\Var\left\{S(\flip)\mid\obsFlip\right\}\right] &= (2^n)^{-1}\sum_{F \in \flipSpace}\left[ \Var\left\{S(\II)\right\}-\Var\left\{S(\obsFlip)\right\}\right] \\
&= n^{-1}\asA^T\mathrm{diag}(\tilde{H}B+B\tilde{H}-\tilde{H}B\tilde{H})\asA+o_{p}(1) \\
&\le n^{-1}\asA^T\mathrm{diag}(\tilde{H}B+B\tilde{H}) \asA+o_{p}(1).
\end{align*}
By Assumption \ref{ass:dispersion}, note that there exists two finite positive constants $c_1, c_2$ such that for each $i$-th element
\begin{equation*}
    2\ash_{ii} b_i \le c_1 b_{sup} \ash_{ii}\le c_2 \sup_{1\le i \le n} \ash_{ii}=c_2  \ash_{sup} 
\end{equation*}
where 
\begin{equation*}
    \ash_{\sup}=\sup_{1\le i \le n} (\ash_{ii}) \quad b_{\sup}=\sup_{1\le i \le n} (b_i).
\end{equation*} 
Therefore we can derive the upper bound
\begin{equation*}
    \E \left[\Var\left\{S(\II)\right\}-\Var\left\{S(\flip)\mid\obsFlip\right\}\right]\le n^{-1}\|\asA\|^2 c_2 \cdot  \ash_{\sup}.
\end{equation*}
Using Assumption \ref{ass:fisherinfo} and Lemma \ref{lem:hspar}, the limiting behavior is
\begin{equation*}
    \lim_{n\xrightarrow{}\infty} n^{-1}\|\asA\|^2 c_2 \cdot \ash_{\sup} = 0.
\end{equation*}
Meanwhile, for two positive constants $c_3,c_4$ we have for each $i$-th element
\begin{equation*}
     -\sum_{k=1}^n \ash^2_{ik}b_k \ge -\sum_{k=1}^n \ash^2_{ik}c_3 b_{\sup} =-c_3 b_{\sup}\ash_{ii} \ge -c_4 \ash_{\sup}.
\end{equation*}
Then using again Lemma \ref{lem:diag} we can derive the lower bound
\begin{align*}
\E \left[\Var\left\{S(\II)\right\}-\Var\left\{S(\flip)\mid\obsFlip\right\}\right] &\ge n^{-1}\asA^T\mathrm{diag}(-\asH B \asH) \asA+o_{p}(1)\\
&\ge -n^{-1}\|\asA\|^2 c_4 \cdot  \ash_{\sup}+o_{p}(1) \, 
\end{align*}
where, using Assumption \ref{ass:fisherinfo} and Lemma \ref{lem:hspar}, the limiting behavior is
\begin{equation*}
    \lim_{n\xrightarrow{}\infty} -n^{-1}\|\asA\|^2 c_4 \cdot \ash_{\sup} = 0.
\end{equation*}
According to the law of total variance we have
\[
\Var\{S(\flip)\} = \Var\left[\E\{S(\flip)\mid \obsFlip\}\right] + \E\left[\Var\{S(\flip)\mid \obsFlip\}\right].
\]
We know that $\E\{S(\flip)\mid \obsFlip\}$ does not depend on $\flip$, which means that $\Var[\E\{S(\flip)\mid \obsFlip\}]=0$, so $\Var\{S(\flip)\} = \E[\Var\{S(\flip)\mid \obsFlip\}].$ It follows that, marginally over $\obsFlip$,
\[
\lim_{n\xrightarrow{}\infty} \Var\left\{S(I)\right\}-\Var\left\{S(\flip)\right\} = 0.
\]
The same argument given in the last part of the proof of theorem \ref{theo:asympt_eff} applies to deduce the asymptotic exactness of the test considered.
\end{proof}

\subsection{Proof of Proposition \ref{prop:mult_standard}}
\begin{prop*}
Assume that the variances are correctly specified, that is, $\asV=V$, and that Assumptions \ref{ass:dispersion}, \ref{ass:fisherinfomult} and \ref{ass:lindebergmult} hold. The $d$-dimensional standardized sign-flip score vector is finite sample second-moment null-invariant. The test that rejects $H_0$ if \eqref{example.mult} holds is an asymptotically $\alpha$ level test.
\end{prop*}
\begin{proof}
The proof is analogous to Theorem 3 of \cite{Hemerik.etal:2020}.
\end{proof}

\subsection{Proof of Theorem \ref{theo:mult_robust}}
\begin{thm*}
Assume that the variances are misspecified, that is, $V \ne \asV$ and that Assumptions \ref{ass:dispersion},\ref{ass:fisherinfomult},\ref{ass:lindebergmult} hold. For $n\xrightarrow{} \infty$, the $d$-dimensional standardized sign-flip score vector is asymptotically second-moment null-invariant. The test that rejects $H_0$ if \eqref{example.mult} holds is an asymptotically $\alpha$ level test.
\end{thm*}
\begin{proof}
It is trivial to observe that the expected value is a $d$-dimensional zero vector.
We focus on the variance, which now is a $d\times d$ matrix. We prove the robustness elementwise.

The proof for the diagonal elements follows from Theorem \ref{theo:robust}.
Then, we focus on one covariance term of the covariance matrix of the flipped effective score vector. Let $X=(X_1,\dots,X_d)^T$ and
\begin{align*}
    \asA_1&= (I-\asH)\asW^{1/2}X_1 \\
    \asA_2&=(I-\asH)\asW^{1/2}X_2.
\end{align*}
When the variances are misspecified
\begin{align*}
    &\mbox{cov}\{S^1(F),S^2(F)\}=\\
    &\; =n^{-1}\asA_1^TF(I-\asH)\asV^{-1/2}\E\left\{(y-\mu)(y-\mu)^T\right\}\asV^{-1/2}(I-\asH)F\asA_2+o_{p}(1)\\
    &\;=n^{-1}\asA_1^TF(I-\asH)B(I-\asH)F\asA_2+o_{p}(1).
\end{align*}
Take the difference
\begin{align*}
&\mbox{cov}\{S^1(I),S^2(I)\}-\mbox{cov}\{S^1(F),S^2(F)\}=\\
&\;=n^{-1}\asA_1^T\left\{B-F(I-\asH)B(I-\asH)F\right\}\asA_2+o_{p}(1)
\\
&\;=n^{-1}\asA_1^T\left\{F(\asH B+B\asH-\asH B \asH)F\right\}\asA_2+o_{p}(1).
\end{align*}

From this point the asymptotic second-moment null-invariance can be derived directly from the proof of the Theorem \ref{theo:robust}, by applying the same reasoning to each covariance term, replacing Assumptions \ref{ass:fisherinfo}-\ref{ass:lindeberg} with Assumptions \ref{ass:fisherinfomult}-\ref{ass:lindebergmult}.
The asymptotic exactness of the test follows from the proof of Proposition \ref{prop:mult_standard}.
\end{proof}

\subsection{RNA-seq data: summary statistics}
\riccardo{The dataset refers to 344 patients. A total of 20,119 genes are analyzed, and three covariates are included in the model. Here some summary statistics of the variables considered.}

\begin{table}[h]
        \centering
        \begin{tabular}{||c|c||}
            \hline Stage I & Other \\ \hline 
            170 & 174 \\
            \hline
        \end{tabular}
        \caption{Distribution of the pathological stage of the patients}
        \label{tab:pat_stage}
\end{table}

\begin{table}[h]
        \centering
        \begin{tabular}{||c|c||}
            \hline Female & Male \\ \hline 
            110 & 234 \\
            \hline
        \end{tabular}
        \caption{Distribution of the gender of the patients}
        \label{tab:gender}
\end{table}

\begin{figure}[!ht]
\includegraphics[width=12cm]{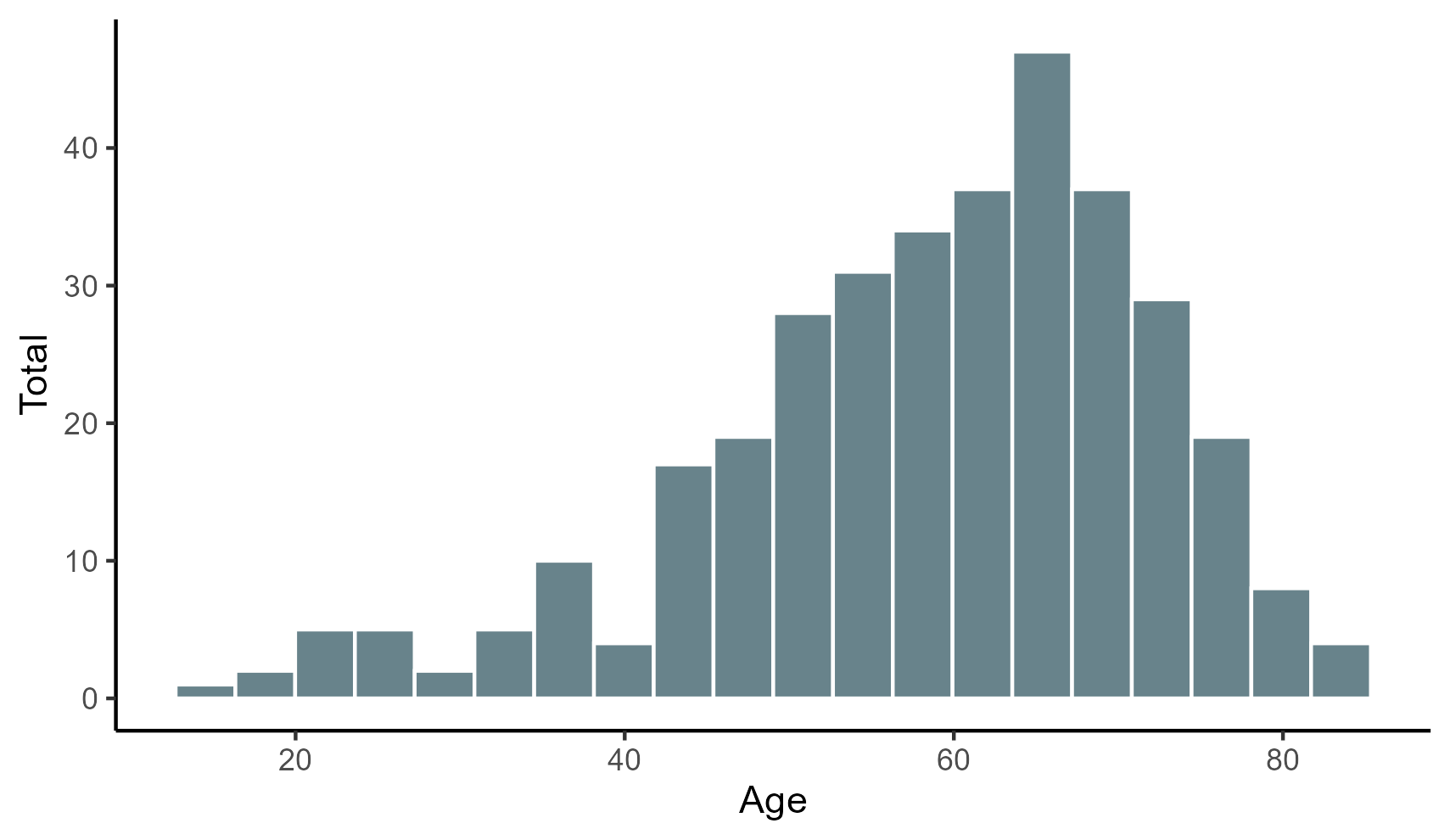}
\caption{Distribution of the age of the patients}
\label{fig:age}
\end{figure}

\begin{table}[h]
        \centering
        \begin{tabular}{||c|c|c||}
            \hline  & Female & Male \\ \hline Stage I& 50 & 120 \\Other & 60 & 114 \\ \hline
        \end{tabular}
        \caption{Bivariate distribution of gender and pathological stage of the patients}
        \label{tab:gender_stage}
    \end{table}

\begin{table}[h]
        \centering
        \begin{tabular}{||c|c|c||}
            \hline  & Stage I & Other \\ \hline (15,29]& 2 & 13 \\(29,43] & 14 & 12 \\ 
            (43,57]& 44 & 50 \\
            (57,71]& 79 & 75 \\
            (71,85]& 31 & 24 \\ \hline
        \end{tabular}
        \caption{Bivariate distribution of age and pathological stage of the patients}
        \label{tab:age_stage}
    \end{table}

\begin{table}[h]
        \centering
        \begin{tabular}{||c|c|c||}
            \hline  & Female & Male \\ \hline (15,29]& 7 & 8 \\(29,43] & 6 & 20 \\ 
            (43,57]& 27 & 67 \\
            (57,71]& 43 & 111 \\
            (71,85]& 27 & 28 \\ \hline
        \end{tabular}
        \caption{Bivariate distribution of age and gender of the patients}
        \label{tab:age_gender}
    \end{table}

\end{document}